\documentclass{article}

\usepackage{microtype}
\usepackage{graphicx}
\usepackage{subfigure}
\usepackage{amsmath,graphicx}
\usepackage{times}
\usepackage{epsfig}
\usepackage{graphicx}
\usepackage{amsmath,amsthm}
\usepackage{amssymb}
\usepackage{placeins}
\usepackage{enumitem}

\usepackage{algorithm}
\usepackage[noend]{algpseudocode}

\newtheorem{theorem}{Theorem}

\newtheorem{definition}[theorem]{Definition}

\newcommand{\Z}{\mathbb{Z}}

\newcommand{\R}{\mathbb{R}}

\newcommand{\norm}[1]{\left\|#1\right\|}
\usepackage{booktabs} 

\usepackage{hyperref}

\usepackage[accepted]{icml2019}

\icmltitlerunning{DAEs for Linear Inverse Problems: Improved Recovery with Provable Guarantees}

\begin{document}

\twocolumn[
\icmltitle{DAEs for Linear Inverse Problems: Improved Recovery with Provable Guarantees}




\begin{icmlauthorlist}
\icmlauthor{Jasjeet Dhaliwal}{to}
\icmlauthor{Kyle Hambrook}{to}

\end{icmlauthorlist}

\icmlaffiliation{to}{Department of Mathematics, San Jose State University, San Jose, California}

\icmlcorrespondingauthor{Jasjeet Dhaliwal}{jasjeet.dhaliwal@sjsu.edu}
\icmlcorrespondingauthor{Kyle Hambrook}{kyle.hambrook@sjsu.edu}

\icmlkeywords{Compressive Sensing, Linear Inverse Problems, Generative Prior}

\vskip 0.3in
]



\printAffiliationsAndNotice{}  

\begin{abstract}
   Generative priors have been shown to provide improved results over sparsity priors in linear inverse problems. However, current state of the art methods suffer from one or more of the following drawbacks: (a) speed of recovery is  slow;  (b) reconstruction quality  is deficient; (c) reconstruction quality is contingent on a computationally expensive process of tuning hyperparameters. In this work, we address these issues by utilizing Denoising Auto Encoders (DAEs) as priors and a projected gradient descent algorithm for recovering the original signal. We provide rigorous theoretical guarantees for our method and experimentally demonstrate its superiority over existing state of the art methods in compressive sensing, inpainting, and super-resolution. We find that our algorithm  speeds up recovery by two orders of magnitude (over 100x),  improves quality of reconstruction by an order of magnitude (over 10x), and does not require tuning hyperparameters. 
   
\end{abstract}
\section{Introduction}
\label{introduction}
Linear inverse problems can be formulated mathematically as
\begin{equation}
    y = Ax + e \nonumber
\end{equation} 
where  $y \in \R^m$ is the observed vector, $A \in \R^{m \times N}$ is the measurement process, $e \in \R^m$ is a noise vector, and $x \in \R^N$ is the original signal.  The problem is to recover the signal $x$, given the observation $y$ and the measurement matrix $A$. Such problems arise naturally in a wide variety of fields including image processing, seismic and medical tomography, geophysics, and magnetic resonance imaging. In this paper, we focus on three linear inverse problems encountered in image processing:  compressive sensing, inpainting, and super-resolution. We motivate our method using the compressive sensing problem. \\

\begin{figure}[h]
\centering
\includegraphics[width=7.5cm, height=7.5cm]{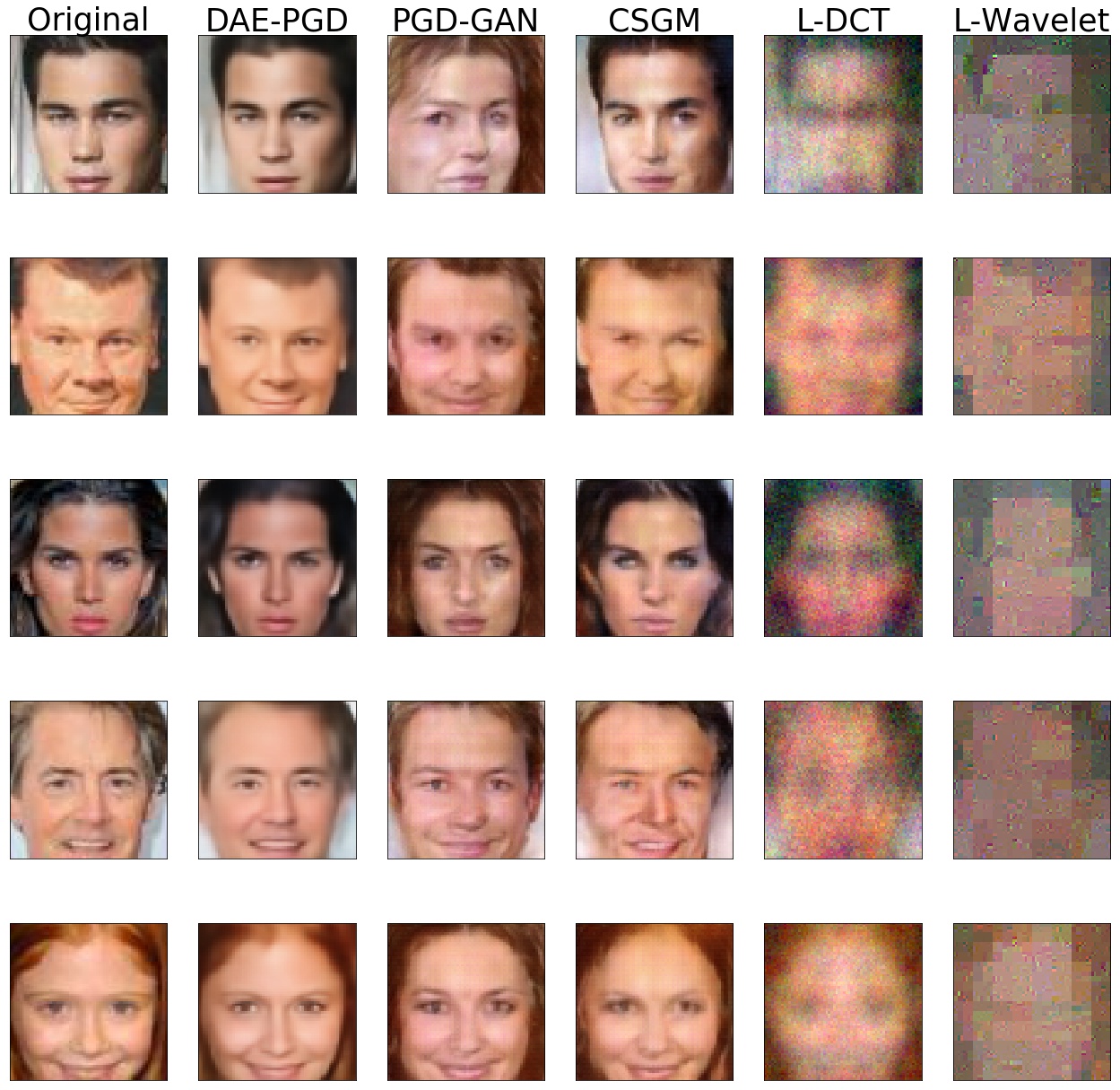}
\setlength{\belowcaptionskip}{-10pt}
\caption{Compressive sensing (w/out noise) CelebA for $m=1000$. DAE-PGD reconstructions show a 10x improvement in quality.}
\label{fig:cs_no_noise_1000}
\end{figure}

\noindent \textbf{Sparsity Prior} 
The problem of compressive  sensing  assumes the matrix $A \in \R^{m \times N}$ is fat, i.e. $m < N$.  Even when no noise is present ($y = Ax$), the system is under determined and the recovery problem is intractable. However, it has been shown that if the matrix  $A$ satisfies certain conditions such as the Restricted Isometry Property (RIP)  and if $x$ is known to be approximately sparse in some fixed basis, 
then $x$ can typically be recovered 
even when $m \ll N$ \cite{tibshirani1996regression, donoho2006compressed, candes2006stable}. 

However, sparsity (or approximate sparsity) is a very restrictive condition 
to impose on the signal as it limits the applicability of recovery methods to a small subset of input domains. In order to ease this constraint, there has been considerable effort in using other forms of structured priors such as structured sparsity \cite{baraniuk2010model}, sparsity in tree-structured dictionaries \cite{peyre2010best}, and  low-rank mixture of Gaussians \cite{chen2010compressive}. Although these efforts improve on the sparsity prior, they do not cater to signals that are not naturally sparse or structured-sparse. \\

\noindent \textbf{Generative Prior} 
Bora et al. \cite{bora2017compressed} address this issue by replacing the sparsity prior on $x$ with a generative prior. In particular, the authors first train a generative model $f: \R^k \mapsto \R^N$ with $k < N$ that maps a lower dimensional latent space to the higher dimensional ambient space. This model is referred to as the generator. Next, they impose the prior that the original signal $x$ lies in (or near) the range of $f$. Hence, the recovery problem reduces to finding the best approximation to $x$ in $f(\R^{k})$. 

It is crucial to note that the quality of the generative prior depends on how well the training set captures the data distribution. Bora et al.\cite{bora2017compressed} used a Generative Adversarial Network (GAN) as the generator, $G: \R^k \mapsto \R^N $, where $k < N$, to model the distribution of the training data and posed the following non-convex optimization problem  
\begin{equation}
\label{bora_recovery}
    \hat{z} = \underset{z \in \R^k}{\text{arg min }} \left( \|AG(z) - y\|^2 + \lambda\|z\|^2 \right)\nonumber  
\end{equation}

\noindent such that $G(\hat{z})$ is treated as the approximation to $x$. The authors provided recovery guarantees for their methods and validated the efficacy of using generative priors by showing that their method required 5-10x fewer measurements than Lasso (with a sparsity constraint) \cite{tibshirani1996regression}  while yielding the same accuracy in  recovery. However, since the problem is non-convex and requires a search over $\R^k$, it is computationally expensive and the reconstruction quality depends on the initialization vector $z \in \R^k$. 

Since then, there have been significant efforts to improve recovery results using neural networks as generative priors \cite{adler2017solving, fan2017inner,gupta2018cnn, liu2017image, mardani2018neural, metzler2017learned,mousavi2017deepcodec, rick2017one, shah2018solving,yeh2017semantic,raj2019gan,heckel2018deep}.  Shah et al. \cite{shah2018solving} extended the work of \cite{bora2017compressed} by training a generator $G$ and using a projected gradient descent algorithm that consists of a gradient descent step $w_t = x_t - \eta A^T(Ax_t - y)$ followed by a projection step $x_{t+1} = G(\underset{z \in \R^k}{\text{arg min}} \|G(z) - w_t\|^2)$. The core idea being that the estimate $w_t$ is improved by projecting it onto the range of $G$. However, since their method requires solving a non-convex optimization problem at every update step, it also leads to slow recovery.

Raj et al. \cite{raj2019gan} enhanced the results of \cite{shah2018solving} by eliminating the expensive non-convex optimization based projection step with one that is an order of magnitude cheaper. In particular, they trained a GAN $G$ to model the data distribution and also trained a  pseudo-inverse GAN $G^{\ddagger}$ that learned a mapping from the ambient space to the latent space. Next, they used the  projection step: $x_{t+1} = G(G^{\ddagger}(w_t))$. By eliminating the need to solve a non-convex optimization problem to update $x_{t+1}$, they were able to attain a significant speed up in the running time of the recovery algorithm. 

However, the recovery algorithm of \cite{raj2019gan} has two main drawbacks. First, training two networks: $G$ and $G^{\ddagger}$ makes the training process and the projection step unnecessarily convoluted. Second, their recovery guarantees only hold when the learning rate $\eta = \frac{1}{\beta}$, where $\beta$ is a RIP-style constant of the matrix $A$. Since it is NP-hard to estimate the constant $\beta$ \cite{bandeira2013certifying}, it follows that setting $\eta = \frac{1}{\beta}$ is NP-hard as well. \footnote{We observed this problem when trying to reproduce the experimental results of \cite{raj2019gan}. Specifically, we tried an exhaustive grid-search for $\eta$ but each value led to poor reconstruction quality.}. \\

\noindent \textbf{DAE Prior} 
In an effort to address the aforementioned issues, we propose to use a DAE  \cite{vincent2008extracting} prior in lieu of the generative prior introduced by Bora et al. \cite{bora2017compressed}. It has previously been shown that DAEs not only capture useful structure of the data distribution \cite{vincent2010stacked} but also implicitly capture properties of the data-generating density \cite{alain2014regularized, bengio2013generalized}. Moreover, as DAEs are trained to remove noise from vectors sampled from the input distribution, they integrate naturally with gradient descent algorithms that lead to noisy approximations at each time step. In consideration of the above, we hypothesize that DAEs are viable candidates for projection operators in a gradient descent based recovery algorithm. 

We therefore replace the generator $G$ used in Bora et al. \cite{bora2017compressed} with a DAE $F: \R^N \mapsto \R^N$ such that the range of $F$ contains the vectors from the original data generating distribution. We then impose the prior that the original signal $x$ lies in the range of $F$ and utilize Algorithm  \ref{PGD_algorithm} to recover an approximation to $x$. We provide theoretical recovery guarantees and find that our framework is able to address the shortcomings of previous works noted above. Our contributions can be summarized as:
   \icmlitem We provide rigorous theoretical guarantees for convergence in Algorithm \ref{PGD_algorithm}.
\icmlitem We experimentally demonstrate orders of magnitude (over 100x) speed up in recovery compared to state of the art methods.
   \icmlitem We experimentally demonstrate order of magnitude (over 10x) improvement in recovery quality compared to state of the art methods.

\section{Algorithm and Results}

\subsection{Notation} Given a vector $x \in \R^N$, we use $\norm{x}$ to denote the $\ell_2$-norm for $x$. Similarly, for a matrix $A \in \R^{m \times N}$, $\norm{A}$ denotes the induced matrix norm from the $\ell_2$-norm. 

\subsection{Denoising Auto Encoder}
\label{dae_intro}
A DAE is a non-linear mapping $F: \R^N \mapsto \R^N$ that can be written  as a composition of two non-linear mappings -  an encoder $E: \R^N \mapsto \R^k$ where $k < N$ and a decoder $D: \R^k \mapsto \R^N$. Therefore,  $F(x) = (D \circ E) (x)$. Given a set of $n$ samples from a domain of interest  $\{x_{i}\}_{i=1}^{n}$, the training set $X$ is created by adding  Gaussian noise to the original samples. That is,  
$X= \{x'_{i} \}_{i=1}^{n}$, where $x'_{i} = x_i + e_i$ and $e_i \sim \mathcal{N}(\mu_i,\,\sigma^{2}_{i})$. 

The loss function for training $F$ is the Mean Squared Error (MSE) loss defined as : $L_{F}(X) =\frac{1}{n} \sum_{i=1}^{n}\|F(x'_i) - x_i\|^2$. The training procedure uses gradient descent to minimize $L_{F}(X)$ with back-propagation.

\subsection{Algorithm}
\label{algo_section}
Recall that in the linear inverse problem 
\begin{equation}
    y = Ax + e \nonumber
\end{equation} 

\noindent our goal is to recover an approximation $\hat{x}$  to $x$ such that $\hat{x}$ lies in the range of $F$.  Thus we aim to find $\hat{x}$ such that 
\begin{equation}
    \hat{x} =  \underset{z \in F(\R^N)}{\text{arg min}} \|Az - y\|^2 \nonumber  
\end{equation}
As in \cite{shah2018solving,raj2019gan}, we use a projected gradient descent algorithm. 
Given an estimate $x_t$ at iteration $t$, 
we compute a gradient descent step for solving the unrestricted problem: $\underset{z \in \R^N}{\text{minimize}} \|Az - y\|^2$ as:
\begin{equation}
w_t \gets x_t -\eta A^{T}(Ax_t - y)\nonumber
\end{equation}
Next we project $w_t$ onto the range of $F$ to satisfy our prior:  \begin{equation}
    x_{t+1} = F(w_t) \nonumber 
\end{equation}
Note that, compared to  \cite{shah2018solving,raj2019gan}, 
the projection step does not require solving a non-convex optimization problem. 



Now suppose that the domain of interest is represented by the set $D \subseteq \R^N$. Then, given a vector $x' = x + e$, where $x \in D$, and $e \in \R^N$ is an unknown noise vector, the success of our method depends on how small the error $\|A(x') - x\|$ is. 
If the training set $X$ captures the domain of interest well and if the training procedure utilizes a diverse enough set of noise vectors $\{e_i\}_{i=1}^{N}$, then we expect $\|A(x') - x\|$ to be small. Consequently,  we expect  the projection step of Algorithm \ref{PGD_algorithm} to yield vectors in or close to $D$. We provide the complete algorithm below. 

\begin{algorithm}
    \caption{DAE-PGD}
    \label{PGD_algorithm}
     \textbf{Input:} $y \in \R^{m}, A \in \R^{m \times N}, f\colon \R^{N} \to \R^{N}$, $ T \in \Z_{+}, \eta \in \R_{> 0} $ \\
     \textbf{Output}: $x_T$ \\
     \begin{algorithmic}[1]
      \State $t\gets 0, x_0 \gets 0$
      \While{$t < T$}
        \State $w_t \gets x_t -\eta A^{T}(Ax_t - y)$
        \State $x_{t+1} \gets f(w_t)$
      \EndWhile\label{euclidendwhile}

    \State \textbf{return} $x_T$
        \end{algorithmic}
    \end{algorithm}

\begin{figure}[h]
\centering
\includegraphics[width=7.5cm, height=7.5cm]{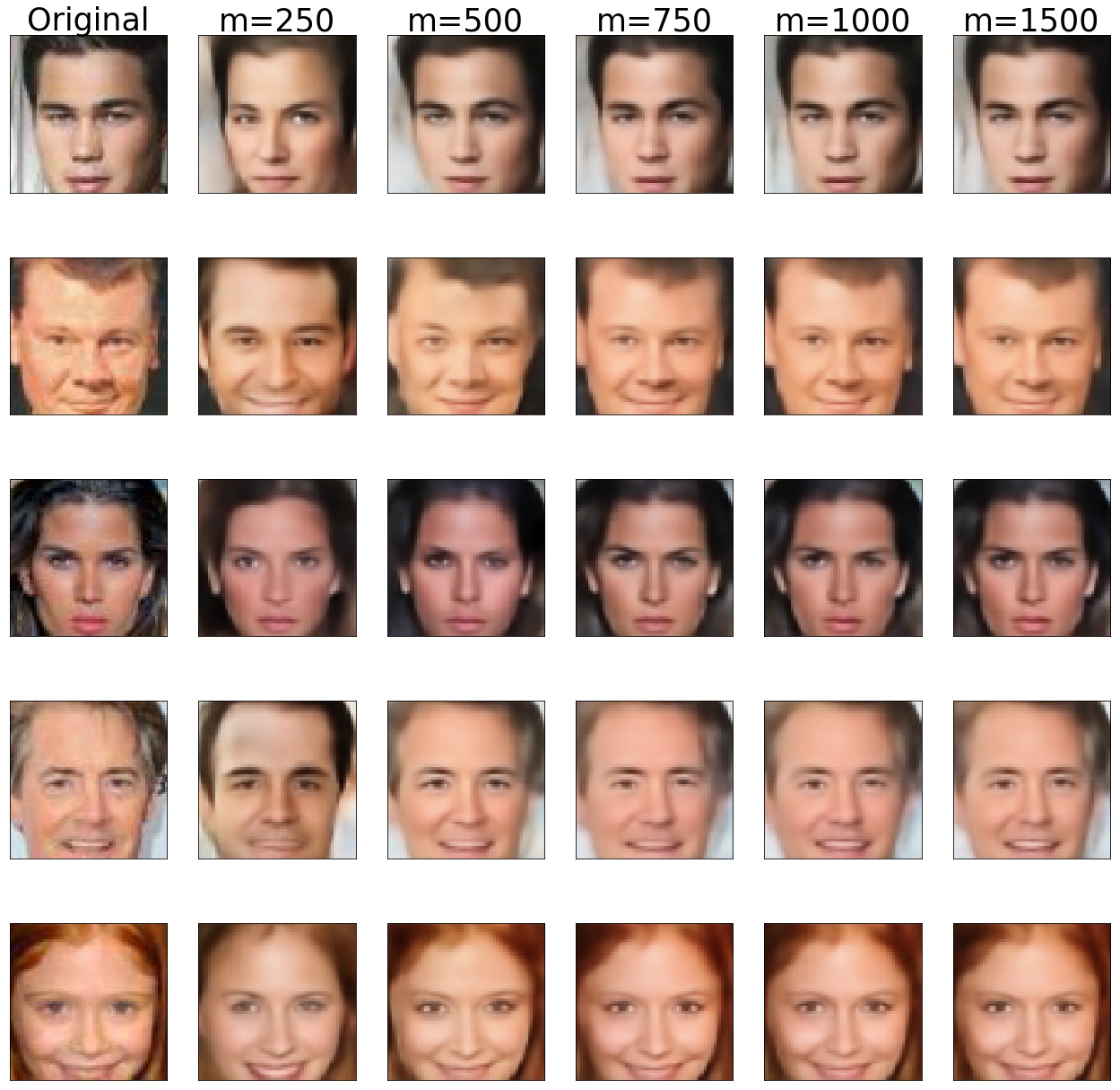}
\caption{Compressive Sensing (w/out noise) CelebA for various $m$. Reconstructions capture finer grained details as $m$ increases. }
\label{fig:cs_no_noise_m}
\end{figure}

\subsection{Theoretical Results}
We begin by introducing two  standard definitions required to provide recovery guarantees.
\begin{definition}[RIP$(S, \delta)$]
\label{S_RIP}
Given $S \subseteq \R^{N}$ and $\delta > 0$, a matrix $A \in \R^{m \times N}$ satisfies the RIP$(S, \delta)$ property if 
\begin{equation}
    (1 - \delta)\norm{x_1 - x_2}^2 \leq \norm{A(x_1 - x_2)}^2 \leq (1 + \delta)\norm{x_1 - x_2}^2 \nonumber
\end{equation}
for all $x_1, x_2 \in S$. 
\end{definition}

A variation of the RIP$(S, \delta)$ property for sparse vectors was first introduced by Candes et al. in \cite{candes2005decoding} and has been shown to be a sufficient condition in proving recovery guarantees using $\ell_1$-minimization methods \cite{foucart2017mathematical}. Next, we define an Approximate Projection (AP) property and provide an interpretation that elucidates its role in the results of Theorem \ref{DAE-PGD_theorem}. \footnote{Various flavors of the AP$(S,\alpha)$ property have been used in previous works, such as Shah et al. \cite{shah2018solving} and Raj et al. \cite{raj2019gan}.}.

\begin{definition}[AP(S, $\alpha$)]
\label{approx_proj}
Let $\alpha \geq 0$.  
A mapping $f \colon \R^{N} \to S \subseteq \R^{N}$ satisfies AP$(S, \alpha)$ 
if
\begin{equation}
    \norm{w - f(w)}^2 \leq \|w - x\|^2 + \alpha \|f(w) - x\| \nonumber
\end{equation}
for every $w \in \R^N$ and $x \in S$.

\end{definition}

    We now explain the significance of Def. \ref{approx_proj}. Let $x^* = \underset{z \in S}{\text{arg min }} \norm{w - z}$ and observe 
   \begin{align}
   \norm{w - f(w)}^2 \leq (\norm{w - x^*} + \norm{f(w) - x^*})^2 \label{triangle}
   \end{align}
   Hence, $\alpha \leq \norm{f(w) - x^*} + 2 \norm{w - x^*}$ is needed to ensure the RHS of Def. \ref{approx_proj} is bounded by the RHS of \eqref{triangle}. In other words, for $\alpha$ to be small, the projection error $\norm{f(w) - x^*}$ as well as distance of $w$ to $S$ need to be small. Since the DAE $F$  learns to minimize $\norm{F(w) - x^*}^2$ (Section \ref{dae_intro}), we expect a small projection error. Moreover, if the image of $F$ approximates the data distribution well,  we  expect a small value for $\norm{w - x^*}$ at every gradient descent step of Algorithm \ref{PGD_algorithm}. \footnote{A small value for $\alpha$ is verified experimentally by using the results of Theorem \ref{DAE-PGD_theorem} and observing small recovery error in our experiments. }. 

\begin{theorem}
\label{DAE-PGD_theorem}
Let $f: \R^{N} \to S \subseteq \R^{N}$ satisfy AP(S, $\alpha$) and let $A \in \R^{m \times N}$ be a matrix with $\|A\|^2 \leq M$ that satisfies RIP$(S, \delta)$. If $y = Ax$ with $x \in S$, the recovery error of Algorithm \ref{PGD_algorithm}  is bounded as:

\begin{equation}
\label{DAE-PGD_equation}
    \norm{x_T - x} \leq (2\gamma)^{T} \norm{x_{0} - x} + \alpha 
    \left( \frac{1 - (2\gamma)^{T}}{1 - (2\gamma)}\right)  
\end{equation}

where $\gamma = \sqrt{ \eta^2M(1 + \delta) +  2\eta(\delta - 1) +1 }$.
\end{theorem}

\begin{proof}[\textbf{Proof of Theorem \ref{DAE-PGD_theorem}}]
Using the notion of Algorithm \ref{PGD_algorithm} and the fact that $f$ satisfies $AP(S, \alpha)$ we have 
\begin{align}
 \norm{(w_t - x) - (x_{t+1} - x)}^2 &\leq \norm{w_t - x}^2 +  \alpha \norm{x_{t+1} - x}. \nonumber 
\end{align}

\noindent Noting $\norm{a - b}^{2} = \norm{a}^2 + \norm{b}^2 - 2\langle a,b\rangle$ and re-arranging terms we get 
\begin{align}
  \norm{x_{t+1} - x}^2 \leq  2\langle (w_t - x), (x_{t+1} - x) \rangle + \alpha \norm{x_{t+1} - x}. \nonumber 
\end{align}

\noindent  Now we expand the inner product using $w_t =x_t - \eta A^{T}(Ax_t - y)$ and $y = Ax$ to get
\begin{align}
\label{angle_inquality}
 \norm{x_{t+1} - x}^2 &\leq  2\langle (I - \eta A^{T}A)(x_t - x), (x_{t+1} - x) \rangle \nonumber \\ &+ \alpha \norm{x_{t+1} - x}.
\end{align}

\noindent  Using the Cauchy–Schwarz inequality we have
\begin{align}
\label{dot_product_bound}
    &|\langle (I - \eta A^{T}A)(x_t - x), (x_{t+1} - x) \rangle| \nonumber \\ &\leq  \norm{(I -  \eta A^{T}A)(x_t - x)} \norm{(x_{t+1} - x)}  
\end{align}

\noindent By setting $u = x_t - x$,  expanding, 
and using the RIP$(S,\alpha)$  property of $A$, we see that
\begin{align}
    \norm{(I - \eta A^{T}A)u}^2 
    = 
    &\norm{u}^2 -2\eta \norm{Au}^2 + \eta^2 \norm{A^T(Au)}^2 \nonumber \\
    \leq &\|u\|^2 - 2\eta (1 - \delta)\|u\|^2 \nonumber \\
    &+  \eta^2 (1+\delta)M\|u\|^2 \nonumber \\
    =
    &\gamma^2 \norm{u}^2 \label{beta_eq}
\end{align}

\noindent We  substitute the results of \eqref{dot_product_bound} and \eqref{beta_eq} into \eqref{angle_inquality} and divide both sides by $\norm{x_{t+1} - x}$ to get  
\begin{align}
\label{inductive_arg}
     \norm{x_{t+1} - x} \leq  2\gamma\norm{x_t - x} + \alpha
\end{align}

\noindent  
Using induction on \eqref{inductive_arg} gives \eqref{DAE-PGD_equation}.
\end{proof}


Theorem \ref{DAE-PGD_theorem} tells us that, if $\gamma < \frac{1}{2}$, then for large $T$, the recovery error is essentially 
$\alpha/(1-2\gamma)$. Note that the requirement $\gamma < \frac{1}{2}$ is satisfied for a large range of values of $\eta$ as long as $\delta$ is sufficiently small \footnote{For instance, random Gaussian matrices yield small values for $\delta$ with high probability \cite{foucart2017mathematical}}. Hence, as long as the value of $\alpha$ is small, we expect to see a small recovery error.


We now compare the above results to Theorem 1 of \cite{raj2019gan} and Theorem 2.2 of \cite{shah2018solving}. As mentioned in Section \ref{introduction}, convergence in Theorem 1 of \cite{raj2019gan} is only guaranteed when $\eta = \frac{1}{\beta}$, which is a much more restrictive condition on $\eta$ than Theorem \ref{DAE-PGD_theorem} provides. In fact, $\beta$ is a RIP-style constant that is NP-hard to find \cite{bandeira2013certifying} which makes setting the value of $\eta = \frac{1}{
\beta}$ NP-hard as well.
Even though the results of Theorem 2.2 from \cite{shah2018solving} require a less restrictive constraint on $\eta$, their guarantees only hold for random Gaussian matrices. Moreover, they require $\norm{A}^2 \leq \omega$, where $\omega$ is a RIP-style constant for $A$. Once again, it is NP-hard to estimate $\omega$, hence making the constraint very strict. In contrast, the results of Theorem \ref{DAE-PGD_theorem} apply to arbitrary matrices that satisfy the RIP-$(S,\delta)$ property and without imposing a strict condition on $\norm{A}^2$.

\begin{figure}[h]
\centering
\includegraphics[width=7.5cm, height=7.5cm]{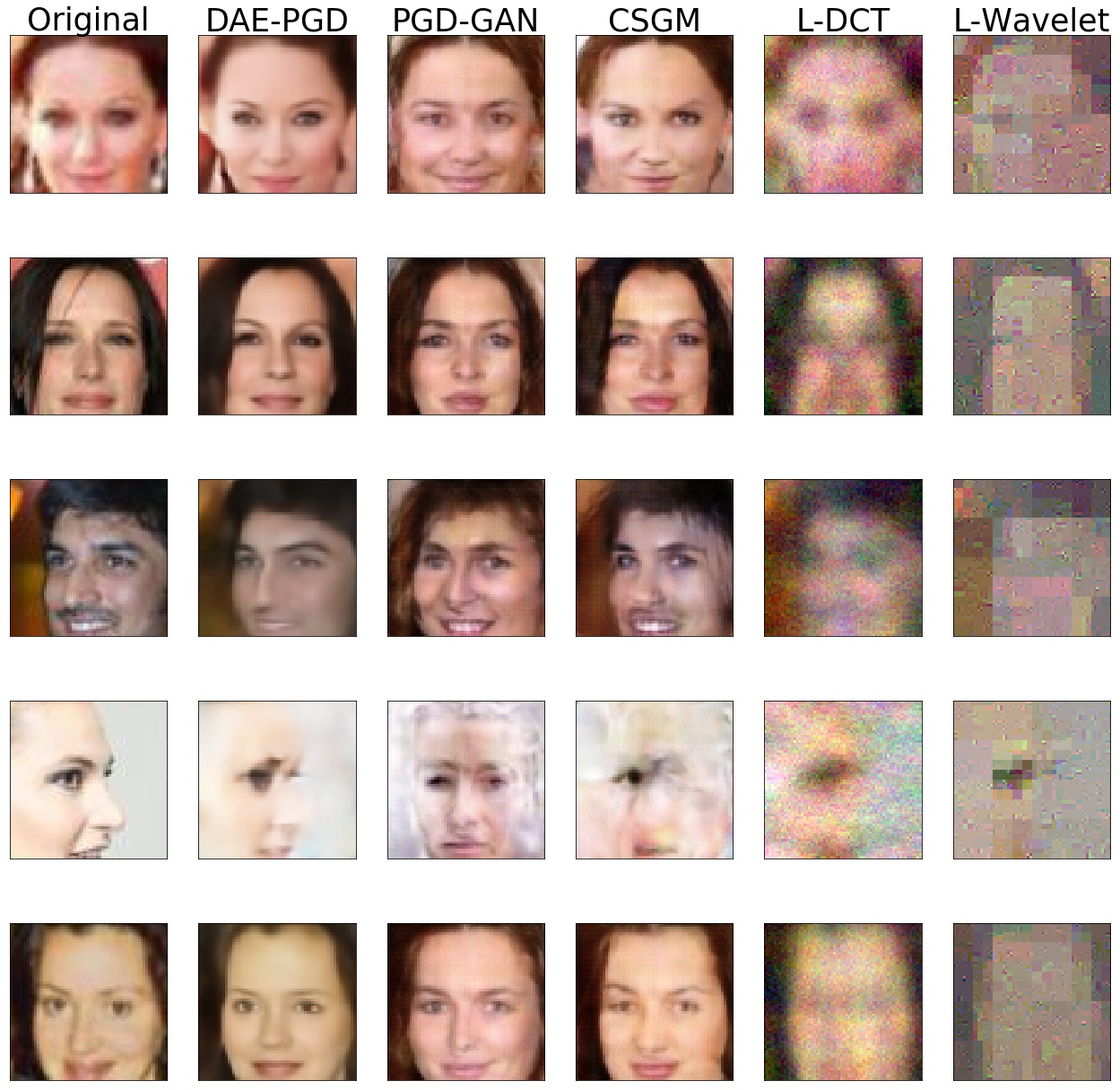}
\caption{Compressive Sensing (with noise) CelebA for $m=1000$. DAE-PGD reconstructions show a 10x improvement in quality.}
\label{fig:cs_noise_1000}
\end{figure}

\section{Experiments}
\begin{figure*}[h]
\begin{tabular}{c c c}
{\includegraphics[width = 5.5cm, height=3.5cm]{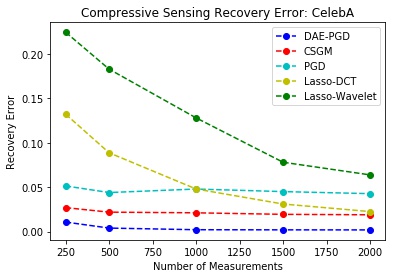}} &
{\includegraphics[width = 5.5cm, height=3.5cm]{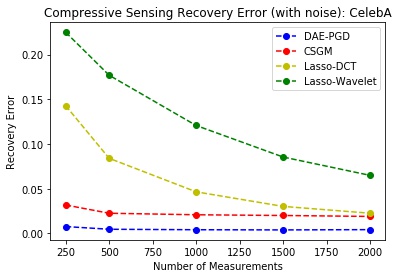}} &
{\includegraphics[width = 5.5cm, height=3.5cm]{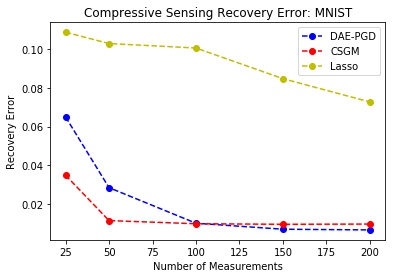}}
\end{tabular}
\caption{Compressive Sensing recovery error: $\norm{x - \hat{x}}^2$. Left: CelebA without noise - DAE-PGD shows over 10x improvement. Middle: CelebA with noise -  DAE-PGD shows over 10x improvement. Right: MNIST without noise - DAE-PGD beats CSGM for $m > 100$.}
\label{fig:recovery_error}
\end{figure*}

We provide experimental results for the problems of compressive sensing,  inpainting, and super-resolution. We refer to the results of Algorithm \ref{PGD_algorithm} as DAE-PGD and compare its results to the methods of Bora et al. \cite{bora2017compressed} which we refer to as CSGM, and  Shah et al. \cite{shah2018solving}\footnote{PGD-GAN results are only provided for compressive sensing on CelebA as per \cite{shah2018solving}.}, which we refer to as PGD-GAN. Although the work of Raj et al. \cite{raj2019gan} is the closest to our method, we do not include comparisons to their work as we were unable to reproduce their results \footnote{We used their code, their trained models, their recovery algorithm, and a grid search for $\eta$ but the reconstructed images were of very poor quality. We also reached out to the authors but they did not have the exact values of $\eta$ that were used in their experiments. }.

\subsection{Setup}

\noindent \textbf{Datasets} 
Our experiments are conducted on the MNIST  \cite{lecun1998mnist} and CelebA \cite{liu2015deep} datasets. The MNIST dataset consists of $28 \times 28$ greyscale images of digits with 50,000 training and 10,000 test samples. We report results for a random subset of the test set. The CelebA dataset consists of more than 200,000 celebrity images. We pre-processes each image to a size of $64 \times 64 \times 3$ and use the first 160, 000 images as the training set and a random subset of the remaining 40,000+ images as the test set. \\

\noindent \textbf{Network Architecture}
The network architectures for our DAEs are inspired by the Variational Auto Encoder architecture from Fig 2. of \cite{hou2017deep} with a few key changes. We replace the Leaky Relu activation with Relu, we  add the two outputs of the encoder to get the latent representation $z$, and we alter the kernel sizes as well as the convolution strides of the network as described in Table \ref{network_arc}.

\noindent \textbf{Training}
We use the Adam optimizer \cite{kingma2014adam} to minimize the  MSE loss function with learning rate 0.01 and a batch size of 128. We train the CelebA network for 400 epochs and the MNIST network for 100 epochs. 

In an effort to ensure that $\|A(x') - x\|$ defined in Section \ref{algo_section} is small, we split the training set into 5 equal sized subsets. For each distinct subset, we sample the noise vectors from a Gaussian distribution $\mathcal{N}(\mu,\,\sigma^{2})$ with a distinct value for $\sigma$ for each subset. The five different values for $\sigma$ that we use are $\{0.25, 0.5, 0.75, 1.0, 1.25\}$. 

All of our experiments were conducted on a Tesla M40 GPU with 12 GB of memory using  Keras \cite{chollet2015} and Tensorflow \cite{tensorflow2015-whitepaper} libraries. The code to reproduce our results is  available  \href{https://github.com/anonymous194758/dae_pgd}{here}.

\subsection{Compressive Sensing}
\label{cs_experiments}
We consider the problem of compressive sensing without noise: $y = Ax$ and with noise: $y = Ax + e$, with $e \sim \mathcal{N}(0, 0.25)$. We  use $m$ to denote the number of observed measurements in our results (i.e. $y \in \R^m$).  As done in previous works \cite{bora2017compressed, shah2018solving, raj2019gan}, the matrix $A \in \R^{m \times N}$ is chosen to be a random Gaussian matrix with $A_{ij} \sim \mathcal{N}(0, \frac{1}{m})$. Finally, we set the learning rate of Algorithm \ref{PGD_algorithm} as $\eta = 1$. Note that in both (with and w/out noise) cases, we also include recovery results for the Lasso algorithm \cite{tibshirani1996regression}  with a DCT basis (L-DCT) and  with a wavelet basis (L-Wavelet).

We begin with CelebA without noise. Figure \ref{fig:cs_no_noise_1000} provides a qualitative comparison of reconstruction results for $m = 1000$. We observe that DAE-PGD provides the best quality reconstructions and is able to reproduce even fine grained details of the original images such as eyes, nose, lips, hair, texture, etc. Indeed the high quality reconstructions support the case that the DAE has a small $\alpha$ as per Def. \ref{approx_proj}. For a quantitative comparison, we turn to Figure \ref{fig:recovery_error} which plots the average squared reconstruction error $\|x - \hat{x}\|^2$ for each algorithm at different values of $m$. Note that DAE-PGD provides more than 10x improvement in the squared reconstruction error. 

In order to capture how the quality of reconstruction degrades as the number of measurements decrease, we add Figure \ref{fig:cs_no_noise_m}, which shows reconstructions for different values of $m$. We observe that even though reconstructions with a small number of measurements capture the essence of the original images, the fine grained details are captured only as the number of measurements increase. 
\begin{table}[t]
\centering
\begin{tabular}{|c| c|c|c|c|c|}
\hline
{\bf m}  & {\bf CGSM} & {\bf PGD-GAN} & {\bf DAE-PGD} &{\bf Speedup}\\
\hline
250 & 53.78 & 48.40& 0.07 & 692x\\
\hline
500 & 59.81& 48.46&0.09 & 538x\\
\hline
1000 &81.08 & 48.46 & 0.11 &440x \\
\hline
1500 & 92.68& 48.50& 0.14 & 346x\\
\hline
2000 & 107.41 & 48.56 & 0.21 &230x \\
\hline
\end{tabular}
\caption{Average running times (in seconds) for the Compressive Sensing problem (w/out noise) on the CelebA dataset.  }
\label{tab:speedup}
\end{table}

\begin{figure*}[h]
\begin{tabular}{c@{\hskip 2cm} c@{\hskip 2cm} c}
\centering
{\includegraphics[width = 4.5cm, height=4.5cm]{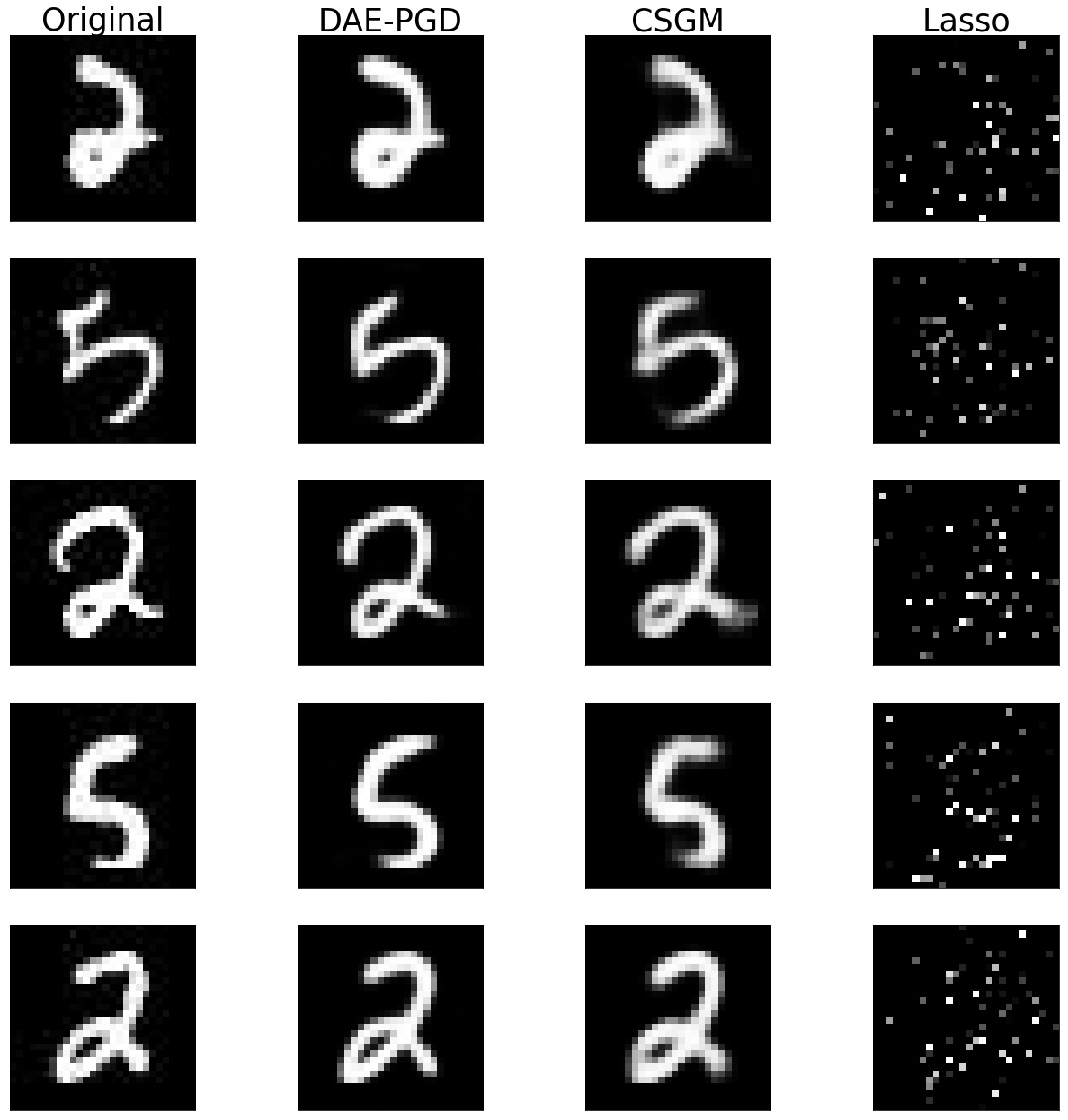}} & 
{\includegraphics[width = 4.5cm, height=4.5cm]{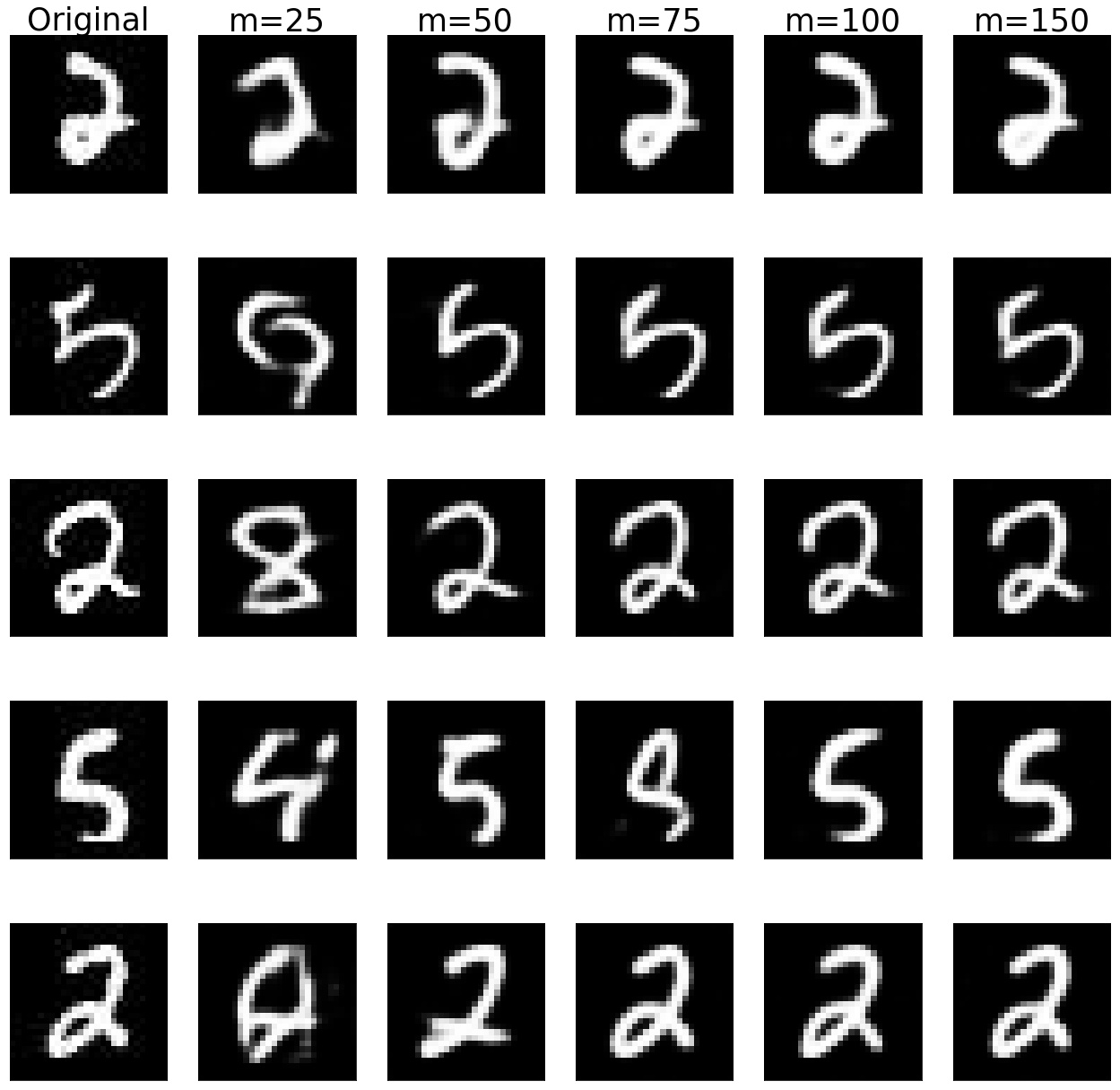}} &
{\includegraphics[width = 4.5cm, height=4.5cm]{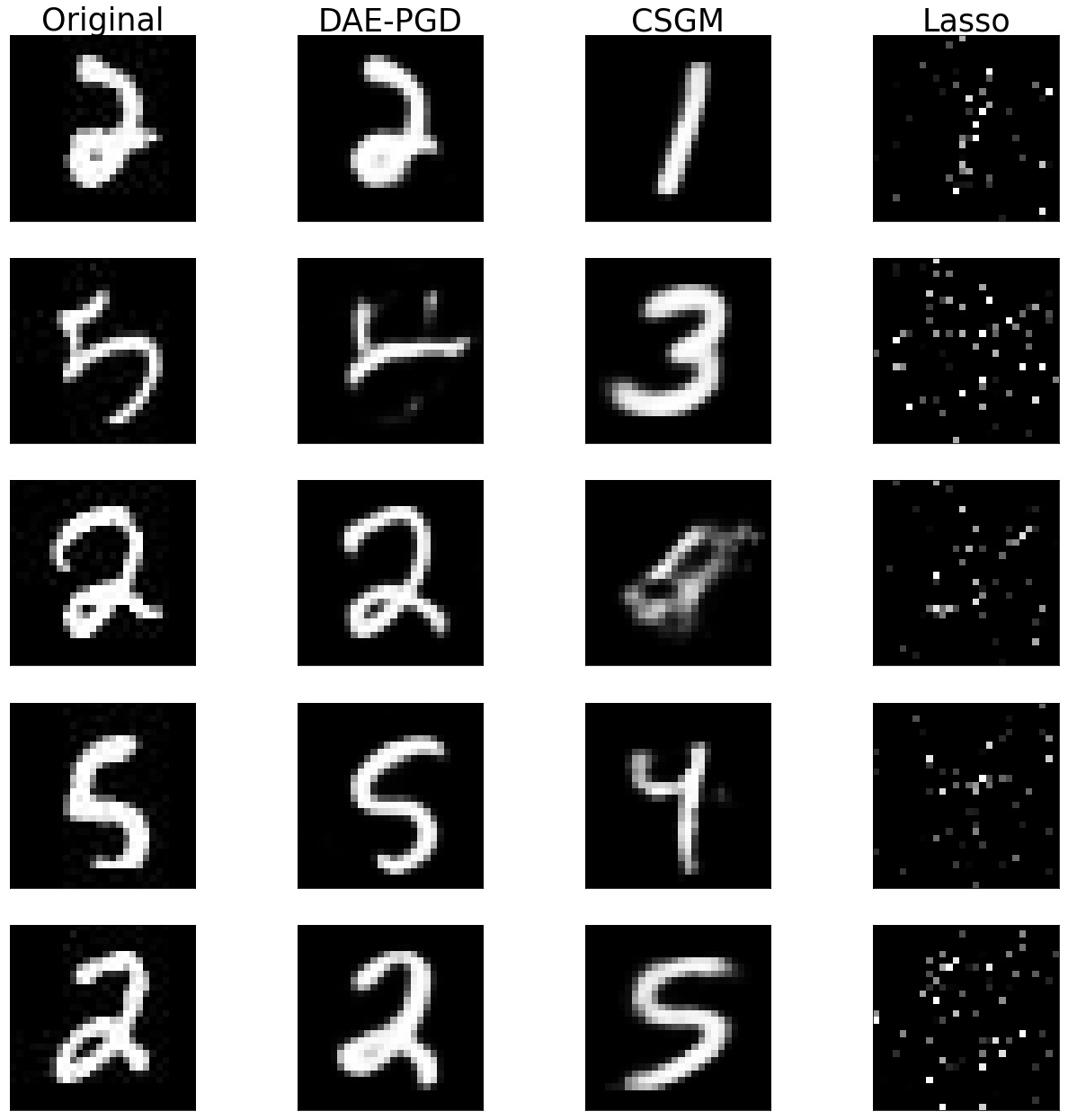}}
\end{tabular}
\caption{Compressive Sensing MNIST. Left: Reconstructions for $m = 100$ without noise. Middle: DAE-PGD reconstructions for different $m$. Right: Reconstructions for $m = 100$ with noise. DAE-PGD and CSGM yield high reconstruction quality for the no-noise case but DAE-PGD outperforms both Lasso and CSGM in the presence of noise. }
\label{fig:mnist_cs}
\end{figure*}

We now turn to the speed of reconstruction. Table \ref{tab:speedup} shows that our method provides speedups of over 100x as compared to PGD-GAN and CSGM \footnote{CSGM is executed for 500 max iterations with 2 restarts and PGD-GAN is executed for 100 max iterations and 1 restart.}.

Next, we provide qualitative reconstruction results for CelebA with additive noise in Figure \ref{fig:cs_noise_1000} note that DAE-PGD clearly outperforms other methods. Moreover, we find that the reconstructions of DAE-PGD once again capture fine-grained details despite the presence of noise in the measurements.  We perform a similar comparison for the MNIST dataset and report results in Figures \ref{fig:recovery_error} and \ref{fig:mnist_cs}.

\subsection{Inpainting}
Inpainting is the problem of recovering the original image, given an occluded version of it. Specifically, the observed image $y$ consists of occluded (or masked ) regions created by applying a pixel-wise mask $A$ to the original image $x$. We use $m$ to refer to the size of mask that occludes a $m \times m$ region of the original image $x$.

We present recovery results for CelebA with $m=10$ in Figure \ref{fig:inpaint_all} and observe that DAE-PGD is able to recovery a high quality approximation to the original image and outperforms CSGM in all cases. Figure \ref{fig:inpaint_all} also captures how recovery is affected by different mask sizes. As in the compressive sensing problem, we find that DAE-PGD reconstructions capture the fine-grained details of each image. Figure \ref{fig:inpaint_all} also reports the result for the MNIST dataset. Even though DAE-PGD outperforms CSGM, we see that the recovery quality of DAE-PGD degrades considerably when $m=15$. We hypothesize this is due to the structure of MNIST images. In particular, since MNIST images are grayscale with most of the pixels being black, putting a $15 \times 15$ black patch on the small area displaying the number makes the reconstruction problem considerably more difficult. This causes considerable degradation in reconstruction quality for larger mask sizes.

\begin{figure*}[thb]
\begin{tabular}{c c c c}
\centering
{\includegraphics[width = 4.2cm, height=4.2cm]{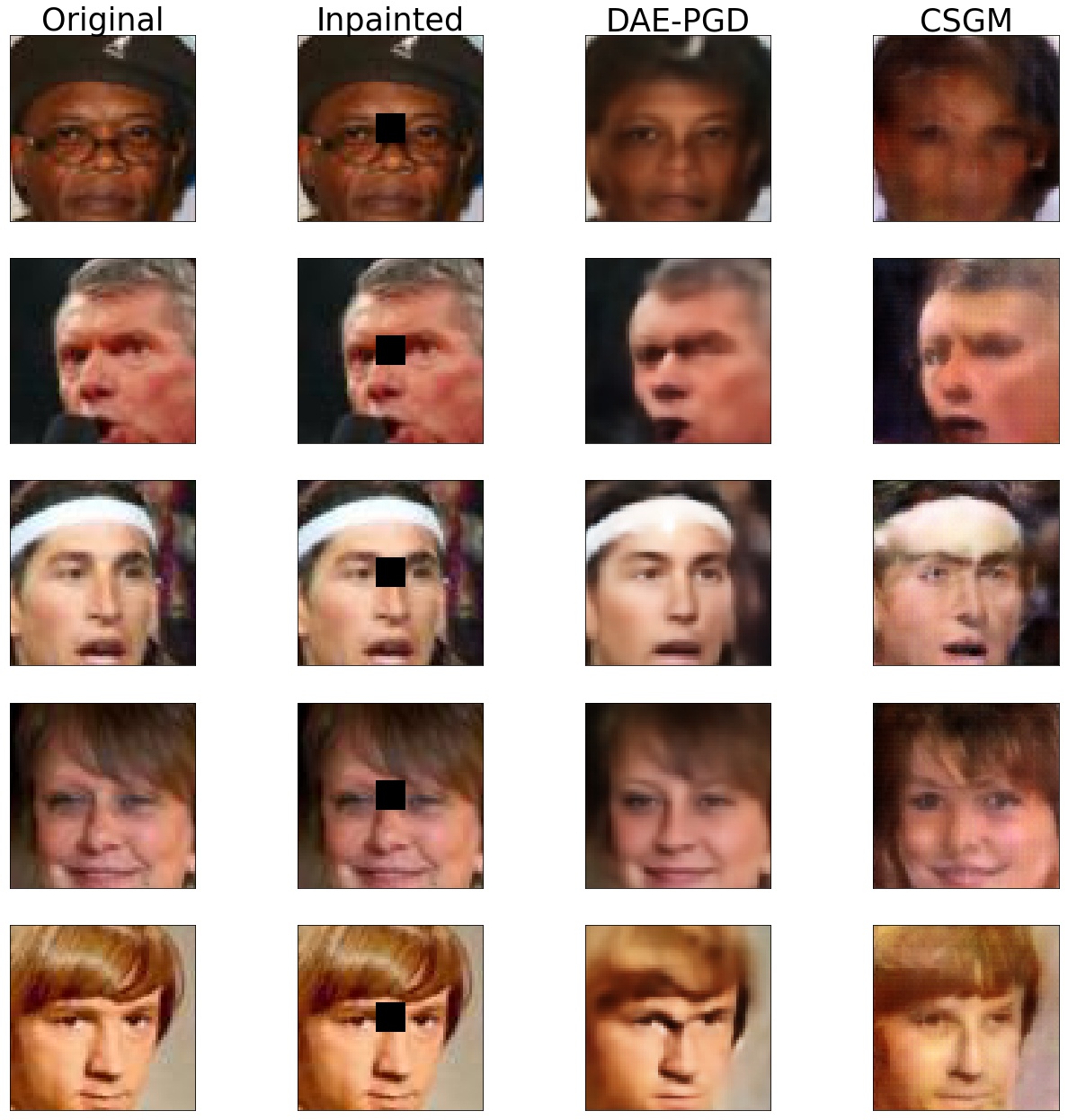}} & 
{\includegraphics[width = 4.2cm, height=4.2cm]{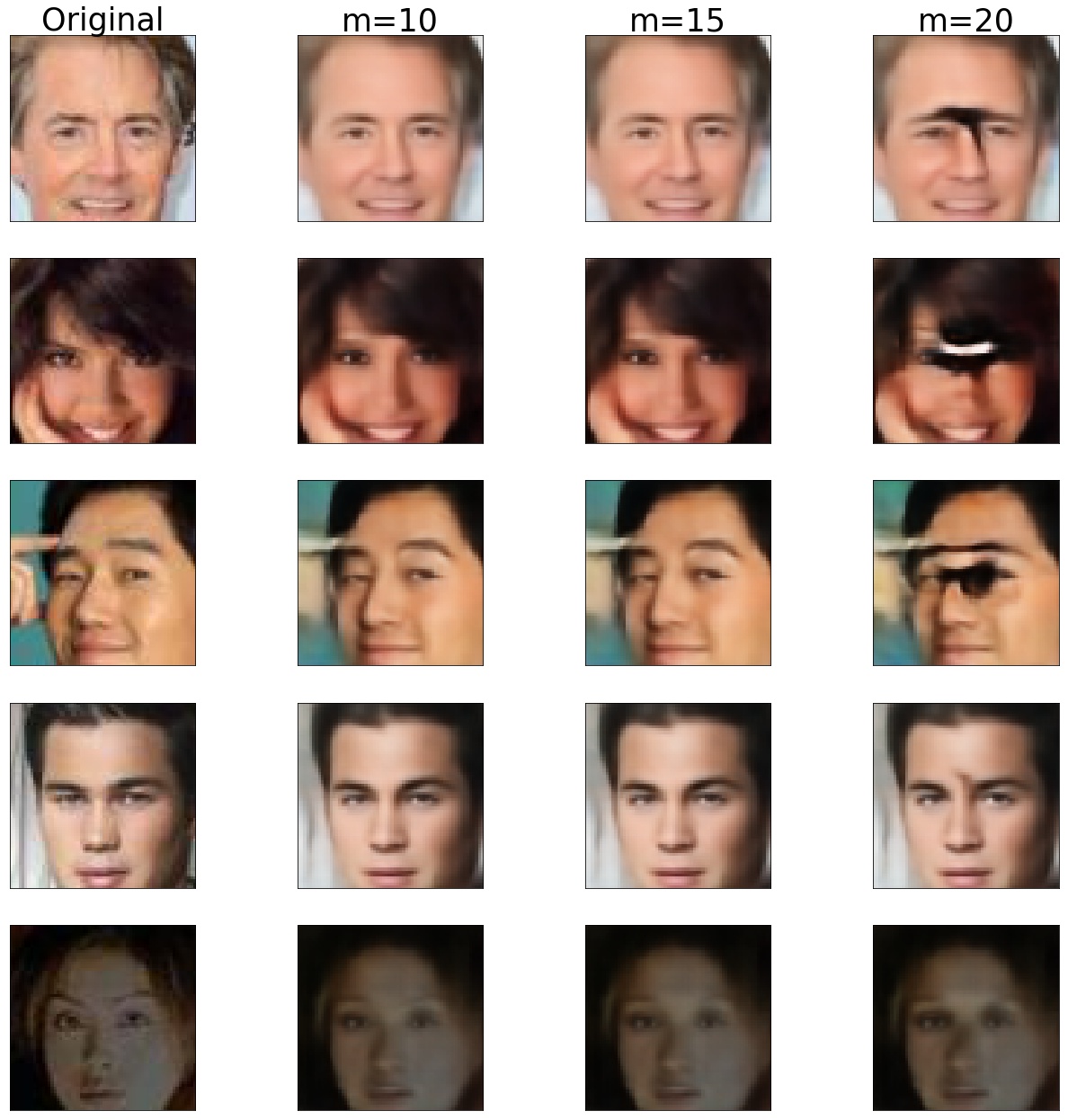}} &
{\includegraphics[width = 4.2cm, height=4.2cm]{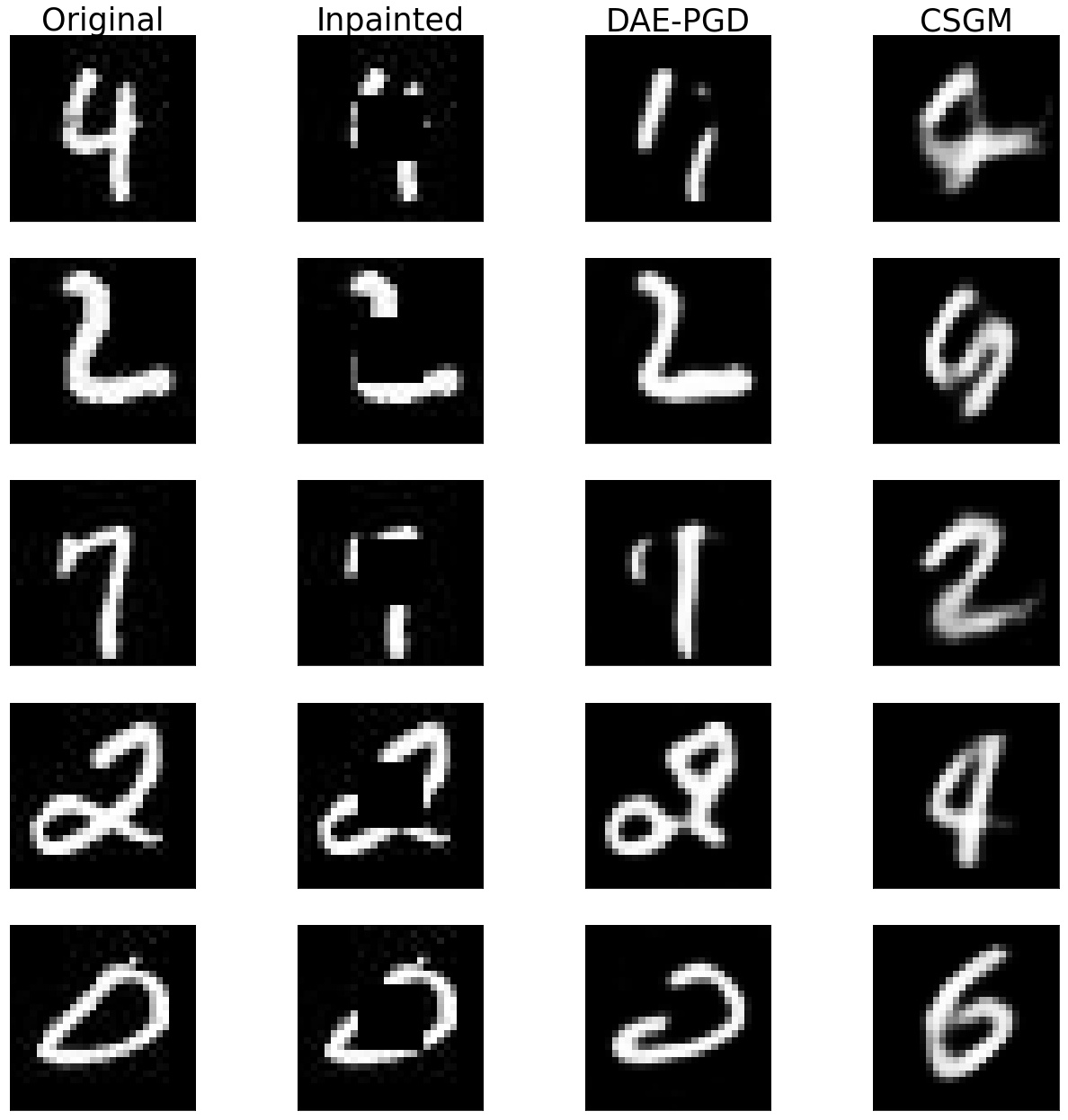}} &
{\includegraphics[width = 4.2cm, height=4.2cm]{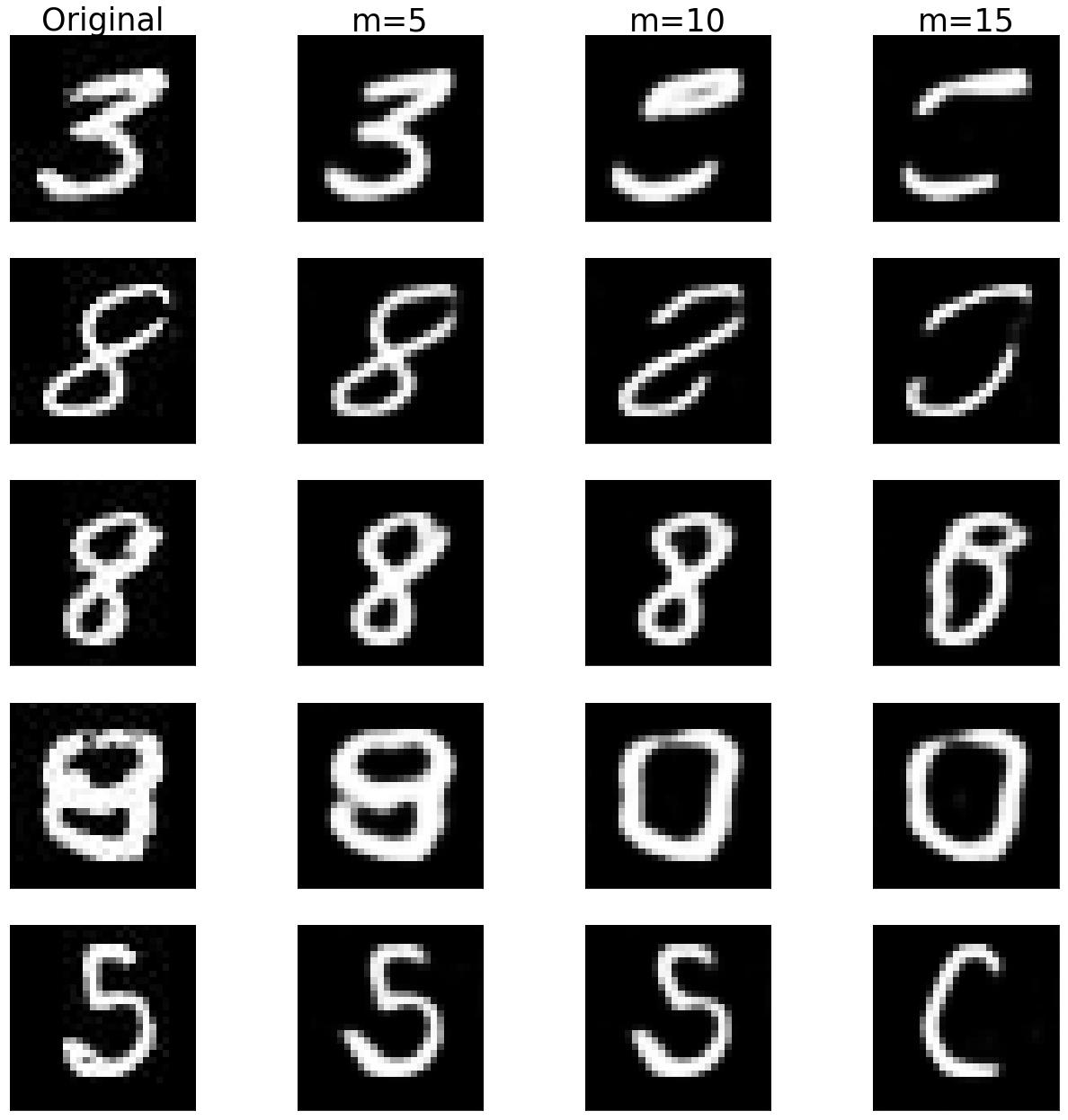} }
\end{tabular}
\caption{Inpainting. Left: CelebA reconstructions for $m = 10$. Middle-Left: DAE-PGD CelebA reconstructions for different $m$. Middle-Right: MNIST reconstructions for $m = 5$. Right: DAE-PGD MNIST reconstructions for different $m$. }
\label{fig:inpaint_all}
\end{figure*}

\subsection{Super-resolution}
Super-resolution is the problem of recovering the original image from a smaller and lower-resolution version. We create this smaller and lower-resolution image by taking the spatial averages of $f \times f$ pixel values where $f$ is the ratio of downsampling. This results in blurring a $ f \times f$ region  followed by downsampling the image. We test our algorithm with $f = 2, 3, 4$ corresponding to $4\times, 9\times$, and $16\times$ smaller image sizes, respectively.

The reconstruction results are provided in \ref{fig:super_all}. We see that DAE-PGD provides higher quality reconstruction for $f=2$ for both CelebA and MNIST. Moreover, reconstruction quality degrades gracefully for CelebA for increasing values of $f$. However, in the case of MNIST, reconstruction quality degrades considerably when $f=4$. Noting that $f=4$ only gives 16 measurements (i.e. $y \in \R^{16})$, we hypothesize that $16$ measurements may not contain enough signal \footnote{Consider compressive sensing with sparsity constraints where recovery guarantees hold when $m \geq Cs \ln (\frac{N}{s})$ \cite{foucart2017mathematical}. } to accurately reconstruct the original images.

\begin{figure*}[thb]
\begin{tabular}{c c c c}
\centering
{\includegraphics[width = 4.2cm, height=4.2cm]{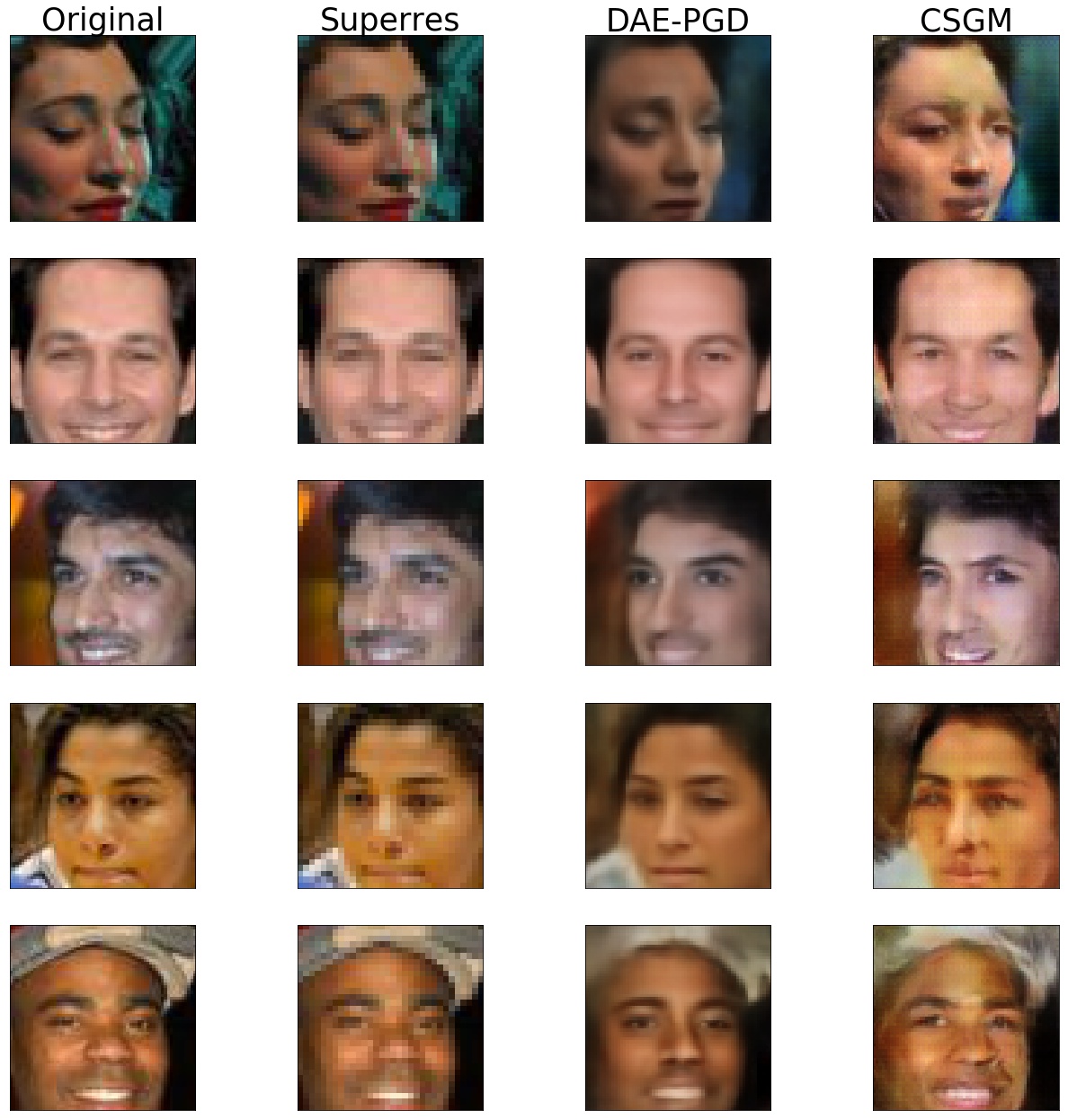}} & 
{\includegraphics[width = 4.2cm, height=4.2cm]{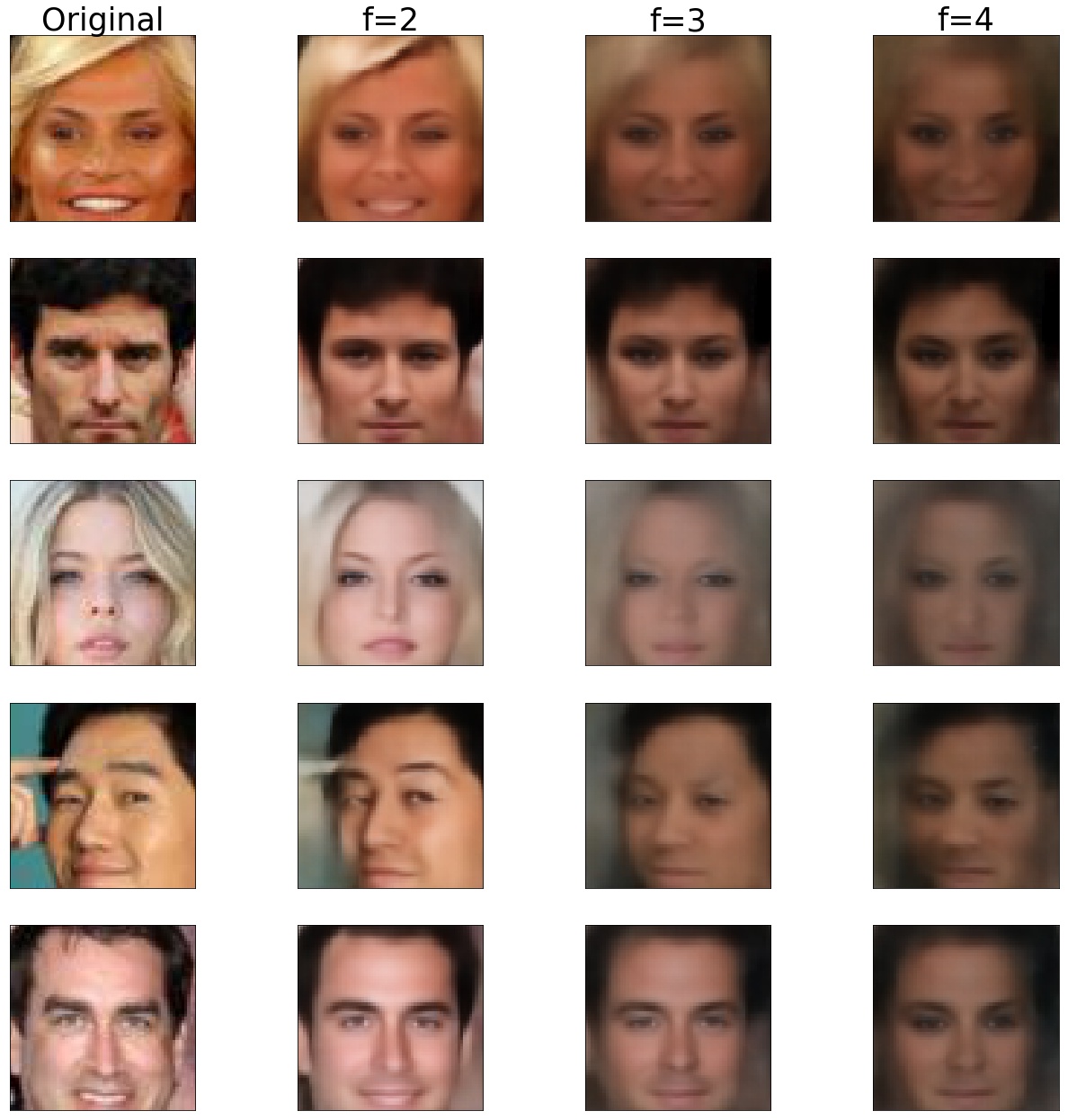}} &
{\includegraphics[width = 4.25cm, height=4.2cm]{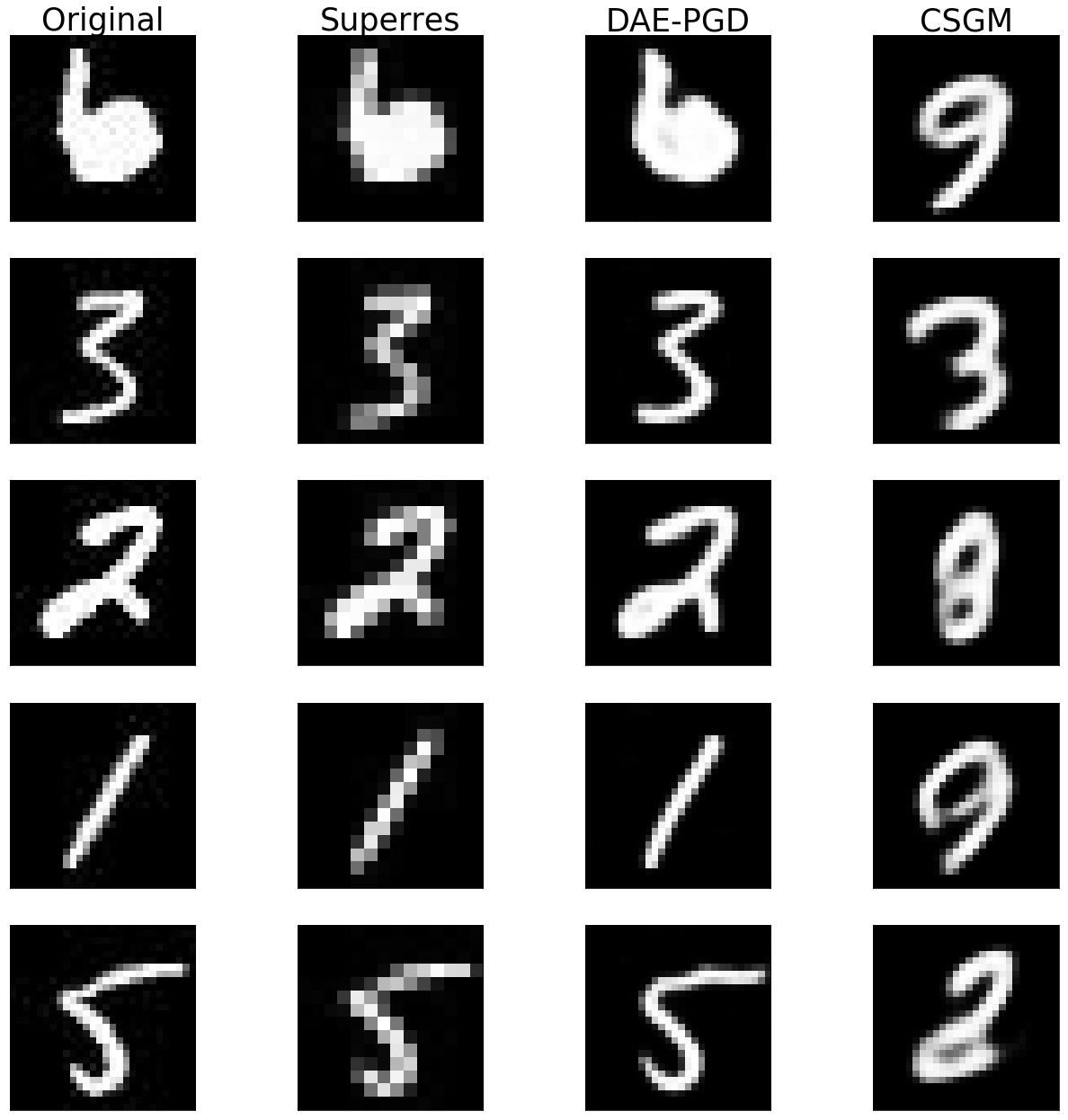}} &
{\includegraphics[width = 4.2cm, height=4.2cm]{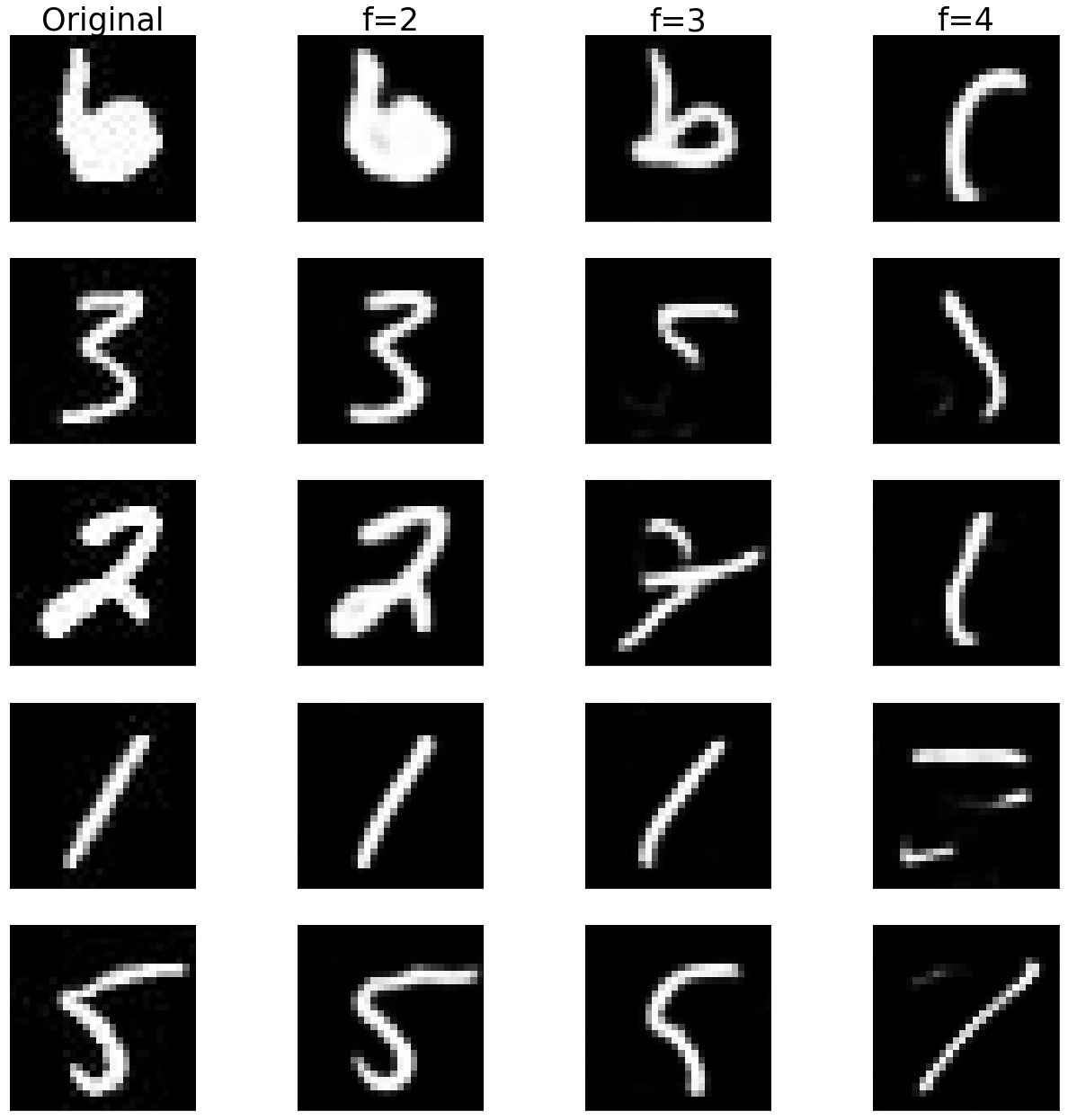} }
\end{tabular}
\caption{Super-resolution. Left: CelebA reconstructions for $f = 2$. Middle-Left: DAE-PGD CelebA reconstructions for different $f$. Middle-Right: MNIST reconstructions for $f = 2$. Right: DAE-PGD MNIST reconstructions for different $f$.}
\label{fig:super_all}
\end{figure*}


\section{Related Work}

\noindent \textbf{Compressive Sensing} The field of compressive sensing was essentially initiated with the work of \cite{candes2006robust} and  \cite{donoho2006compressed} where provided recovery results for sparse signals with a random measurement matrix.
Some of the earlier work in extending compressive sensing  to perform stable recovery with deterministic matrices was done by  \cite{candes2005decoding} and \cite{candes2006stable}, where a sufficient condition for recovery was satisfaction of a restricted isometry hypothesis. \cite{blumensath2009iterative} introduced IHT as an algorithm to recover sparse signals which was later modified in \cite{baraniuk2010model} to reduce the search space as long as the sparsity was structured. 

\noindent \textbf{Generative Priors} Following the lead of \cite{bora2017compressed}, there have been significant efforts to improve on previous recovery results using neural networks as generative models \cite{adler2017solving, fan2017inner,gupta2018cnn, liu2017image, mardani2018neural, metzler2017learned,mousavi2017deepcodec, rick2017one, shah2018solving,yeh2017semantic,raj2019gan,heckel2018deep}. One line of work \cite{jagatap2019algorithmic, heckel2018deep} extends the efforts of Bora et al. \cite{bora2017compressed} by using an untrained neural network $G$ and solving the optimization problem in \eqref{bora_recovery}. However, the optimization problem is highly non-convex and requires a large number of iterations with multiple restarts. Another line of work, \cite{mousavi2017deepcodec,mousavi2017learning} trains a neural network to model the transformation $f(y) = \hat{x}$ where $\hat{x}$ is the approximation to the original input $x$. This approach is limited as a) the inverse mapping is non-trivial to learn and b) will only work for a fixed measurement mechanism. Peng et al \cite{peng2020solving} follow the projected gradient descent method of \cite{shah2018solving} and replace the inner optimization step by two projection steps: 1) mapping the approximation of the gradient descent step into a latent space; 2) mapping the latent space vector back to the original space.

\noindent \textbf{DAEs in Linear Inverse Problems} DAEs have been previously used in image processing tasks such as image denoising \cite{wang2018learning, guo2019agem, Chang_2017_ICCV} and image super-resolution \cite{sonderby2016amortised} to yield good results. However, these approaches  utilize different recovery algorithms and none impose the DAE prior.   Wang et al. \cite{wang2018learning} use gradient descent to minimize the mean shift vector and the mean squared error at each update step. Guo et al. \cite{guo2019agem} utilize an Expectation Maximization style update step at each iteration to recover the original signal. Sonderby et al. \cite{sonderby2016amortised} deploy  Bayes-optimal  denoising to take a gradient step along the log-probability of the data distribution. Chang et. al  \cite{Chang_2017_ICCV} used the alternating direction method of multipliers to solve a Langrangian  formulation that  involves the prior as a constraint.

\section{Conclusion}
We introduced DAEs as priors for general linear inverse problems and provided experimental results for the problems of compressive sensing, inpainting, and super-resolution on the CelebA and MNIST datasets. Utilizing a projected gradient descent algorithm for recovery, we provided rigorous theoretical guarantees for our framework and showed that our recovery algorithm does not impose strict constraints on the learning rate and hence eliminates the need to tune hyperparameters. We compared our framework to state of the art methods experimentally and found that our recovery algorithm provided a speed up of over two orders of magnitude and an order of magnitude improvement in reconstruction quality. 

\begin{table}[t]
\begin{center}
\begin{tabular}{|c|c|c|c|c|}
\hline
{\bf Layer}  & {\bf C-K} & {\bf C-S} &{\bf M-K} & {\bf M-S}\\
\hline
Conv2D 1 & 9 $\times$ 9 & 2 & 5 $\times$  5& 2 \\
\hline
Conv2D 2 & 7 $\times$ 7 &2 &5 $\times$  5 &2 \\
\hline
Conv2D 3 &5$\times$ 5 &  2&3$\times$ 3 &  2\\
\hline
Conv2D 4 & 5 $\times$ 5 & 1& 3$\times$ 3& 1\\
\hline
TransConv2d 1 & 5 $\times$ 5 & 2 & 3$\times$ 3 & 1\\
\hline
TransConv2d 2 & 5 $\times$ 5 &2& 3$\times$ 3 &2 \\
\hline
TransConv2d 3 &7$\times$ 7 &  2&5 $\times$  5&  2\\
\hline
TransConv2d 4 & 9 $\times$ 9 & 1& 5 $\times$  5 & 2\\
\hline
\end{tabular}
\end{center}
\caption{ Network Architectures for CelebA and MNIST. C-K, C-S, M-K, and M-S report CelebA Kernel Sizes, CelebA Strides, MNIST Kernel Sizes, and MNIST strides respectively.}
\label{network_arc}
\end{table}

\bibliography{paper}
\bibliographystyle{icml2019}


%


\end{document}